\documentclass[a4paper, UKenglish, cleveref, autoref, thm-restate]{lipics-v2021}

\title{Profinite trees, through Lawvere theories and the \pdflambda-calculus} 

\author{Vincent Moreau}{Tallinn University of Technology, Estonia }{vincent.moreau@taltech.ee}{https://orcid.org/0009-0005-0638-1363}{(Optional) author-specific funding acknowledgements}

\authorrunning{V. Moreau} 

\Copyright{Vincent Moreau} 



\category{} 

\relatedversion{} 

\acknowledgements{We want to thank Nathanael Arkor, Herman Goulet-Ouellet, Sam van Gool, Keisuke Hoshino, Paul-André Melliès, and Tito Nguy\~{\^e}n for various interesting discussions.}

\nolinenumbers 

\EventEditors{John Q. Open and Joan R. Access}
\EventNoEds{2}
\EventLongTitle{42nd Conference on Very Important Topics (CVIT 2016)}
\EventShortTitle{CVIT 2016}
\EventAcronym{CVIT}
\EventYear{2016}
\EventDate{December 24--27, 2016}
\EventLocation{Little Whinging, United Kingdom}
\EventLogo{}
\SeriesVolume{42}
\ArticleNo{23}

\usepackage{quiver}
\usepackage{macros}
\usepackage{savoirs}

\begin{document}

\maketitle

\begin{abstract}
    The starting point of algebraic language theory is that regular languages of finite words are exactly those recognized by finite monoids. This finiteness condition gives rise to a topological space whose points, called profinite words, encode the limiting behavior of words with respect to finite monoids.

    In this work, we move from words and monoids to trees and clones, the algebraic structures underlying deterministic bottom-up tree automata. Using the categorical notion of codensity monad, we introduce a profinite completion for clones. We prove that this construction on clones simultaneously generalizes the ultrafilter monad on sets and the profinite completion of monoids. When applied to free clones on a ranked alphabet, the profinite completion of clones yields a notion of profinite tree, providing a topological approach to regular languages of finite trees. We prove that these profinite trees coincide with a well-identified fragment of the profinite $\lambda$-calculus.
\end{abstract}

\section{Introduction}

\paragraph*{Algebraic language theory for finite words}

One of the most remarkable properties of the class of regular languages
is its robustness.
Indeed, in the simplest setting consisting of finite words,
regular languages can be described through various equivalent ways,
such as regular expressions,
different notions of finite automata,
and monadic second-order logic.
The field of algebraic language theory
is founded on one such description,
expressed in terms of recognition by finite monoids~\cite[Theorem~1]{Rabin1959}.
One of the benefits of this approach is that
monoid structures are general enough to include both free monoids
\[
    \Sigma^*
    \quad=\quad
    \{\text{words over the alphabet $\Sigma$}\}
\]
and also monoids of transition functions~$Q \To Q$,
whose composition represents the behavior of words when run in deterministic automata.
The monoid homomorphisms
\[
    \delta
    \quad:\quad
    \Sigma^*
    \ \longto\
    Q \To Q
\]
correspond exactly to families of set-theoretic transition functions,
obtained by evaluating the homomorphism on each letter
\[
    \delta(a)
    \quad:\quad
    Q
    \ \longto\
    Q
    \qquad\text{for every letter~$a$ in the alphabet~$\Sigma$}
\]
so we recover deterministic automata, without initial and final states,
in this algebraic setting.
More generally, any monoid homomorphism into a finite monoid
\[
    \delta
    \quad:\quad
    \Sigma^*
    \ \longto\
    M
\]
defines a partition of $\Sigma^*$ of finite index
into regular languages.
Hence, such data can be thought of as a generalized deterministic automaton,
associating to any word $w \in \Sigma^*$ the element~$\delta(w) \in M$ representing the result of the run of $w$.
This approach permits to develop a number of tools of an algebraic nature,
such as the syntactic monoids of languages which are central to Schützenberger's theorem relating star-free languages and aperiodic monoids~\cite{Schtzenberger1965}.

\paragraph*{The profinite syntax of automata theory}

A fundamental property of finite words is that they represent the operations that can be carried in all monoids.
Very concretely, the word $aaba \in \{a, b\}^*$ induces a family of binary operations on all monoids~$M$, defined as
\[
    (x, y)
    \ \in\
    M \times M
    \quad\longmapsto\quad
    x \cdot x \cdot y \cdot x
    \ \in\
    M
\]
Conversely, any family of binary operations defined on all monoids
verifying a naturality condition
is given by a word on the alphabet~$\{a, b\}$.
This demonstrates that finite words provide a syntax for the operations on monoids.

Given that regular languages arise by considering finite monoids,
we can consider the operations which are defined on all finite monoids.
One of them is given by the following fact:
for any finite monoid~$M$, there exists a natural number~$n \ge 1$ such that
for any element~$x \in M$, its power $x^n$ is idempotent, i.e. $x^n \cdot x^n = x^n$.
This property,
akin to an algebraic formulation of the pumping lemma,
defines a natural family of unary operations
\[
    x
    \ \in\
    M
    \quad\longmapsto\quad
    x^n
    \ \in\
    M
\]
on all finite monoids~$M$, which is not given by any finite word.
Nevertheless,
operations defined on all finite monoids correspond to profinite words,
which are the points of the Stone dual space of the Boolean algebra of regular languages, see the work of Pippenger~\cite{Pippenger1997} and Almeida~\cite[\S~3.6]{doi:10.1142/2481}.
Profinite words thus provide a notion of syntax for finite monoids, hence for automata theory,
and they form a compactification of the discrete space of finite words
\[
    \underbrace{\Sigma^*}_{\text{finite words}}
    \quad\subseteq\quad
    \underbrace{\ProMon{\Sigma^*}}_{\text{profinite words}}
\]
The space of profinite words~$\ProMon{\Sigma^*}$ can be built in multiple ways,
as a Cauchy completion,
through Stone duality,
and via naturality,
see the survey~\cite{pin2009}.
Profinite words are of prime importance in the study of finite monoids,
for example they are at the heart of Reiterman's theorem on pseudovarieties of monoids~\cite{Reiterman1982}.
They also appear in related topics such as symbolic dynamics~\cite{Almeida2020} and circuit complexity~\cite{lmcs:4336}.

\paragraph*{Algebraic language theory for finite ranked trees with sharing}

In this article, we apply the algebraic viewpoint on formal language theory
to the case of regular languages of finite trees.
This amounts to moving from the setting of words and automata,
where composition is inherently binary,
to a framework which allows an element to be composed with a finite family of elements following the laws of tree grafting.
For this matter, we identify (abstract) clones,
i.e. families of sets of operations indexed by natural numbers representing their arity,
as an appropriate notion of algebra.
Indeed, these algebras include both free clones
\[
    \Tree(\Sigma)_n
    \quad=\quad
    \{\text{finite ranked trees over~$\Sigma$ with shared variables among $v_1, \dots, v_n$}\}
\]
on a ranked alphabet~$\Sigma$,
and the clones of endofunctions $\Endo(Q)$ on any set~$Q$.
In analogy to the case of monoids, the clone morphisms
\[
    \delta
    \quad:\quad
    \Tree(\Sigma)
    \ \longto\
    \Endo(Q)
\]
then correspond to structures of bottom-up, deterministic automaton without final states on the set $Q$.
The notion of clone is quite flexible,
and in particular we can encode all monoids~$M$ and sets~$X$ as clones
\begin{align*}
    \Act(M)
    \quad & :=\quad
    \text{clone freely generated by the elements of $M$ seen as unary operations}
    \\
    \Cst(X)
    \quad & :=\quad
    \text{clone freely generated by the elements of $X$ seen as constants}
\end{align*}
An important fact about this algebraic setting founded on clones
is that we did not fix a ranked alphabet~$\Sigma$ and did not define our notion of algebra according to it.
Instead, clones provide an abstract notion of algebra,
and we recover a way to deal with ranked alphabets through free clones.

\paragraph*{Profinite completions via codensity monads}

This article aims to establish on firm ground profinite methods in the setting of finite trees.
We start with the simple observation that any monoid has a profinite completion
\[
    \begin{tikzcd}[ampersand replacement=\&]
        M \&\&\&\&\& {\ProMon{M}}
        \arrow["{\text{profinite completion}}", between={0.1}{0.9}, squiggly, from=1-1, to=1-6]
    \end{tikzcd}
    \qquad\qquad\qquad\text{for any monoid~$M$}
\]
and that the monoid of profinite words~$\ProMon{\Sigma^*}$ is the profinite completion of the monoid of finite words~$\Sigma^*$.
This general construction admits a particularly elegant formulation in the language of category theory
through the notion of codensity monad.
The main idea of this abstract construction is that,
under mild hypothesis,
any notion of finite structure induces its associated profinite completion.
In particular,
a theorem by Kennison and Gildenhuy~\cite{KENNISON1971317} states that considering finite sets leads to the ultrafilter monad
\[
    \begin{tikzcd}[ampersand replacement=\&]
        X \&\&\&\&\& {\ProSet X}
        \arrow["{\text{ultrafilter monad}}", between={0.1}{0.9}, squiggly, from=1-1, to=1-6]
    \end{tikzcd}
    \qquad\qquad\qquad\text{for any set~$X$}
\]
which is central to model theory and logic in general.

Well-informed of this categorical apparatus,
we identify a class of clones that we consider as finite algebras
which automatically induces a profinite completion of clones
\[
    \begin{tikzcd}[ampersand replacement=\&]
        C \&\&\&\&\& {\ProClone{C}}
        \arrow["{\text{profinite completion}}", between={0.1}{0.9}, squiggly, from=1-1, to=1-6]
    \end{tikzcd}
    \qquad\qquad\qquad\text{for any clone~$C$}
\]
defined in the same spirit as the natural operations on all finite monoids,
albeit transposed appropriately to the setting of clones.
In particular, this profinite completion yields a notion of profinite tree
defined as the operations of the completion of the associated clone of trees
\[
    \ProClone{\Tree(\Sigma)}_n
    \quad=\quad
    \{\text{profinite trees on $\Sigma$ with shared variables among $v_1, \dots, v_n$}\}
\]
Moreover, these sets naturally carry a structure of Stone space
which is compatible with the algebraic structure of clone.

Having introduced in this article this new construction on clones,
we wish to relate it to other profinite completions.
For that matter, we then prove that the encodings of monoids and sets as clones commute with their profinite completions, i.e.
\begin{align*}
    \Act\big(\ProMon{M}\big)
    \quad & \cong\quad
    \ProClone{\Act(M)}
    \qquad\qquad\qquad\text{for any monoid~$M$}
    \\
    \Cst\big(\ProSet X\big)
    \quad & \cong\quad
    \ProClone{\Cst(X)}
    \qquad\qquad\qquad\text{ for any set~$X$}
\end{align*}
This demonstrates that the profinite completion of clones faithfully generalizes the ones of monoids and sets,
leading to new connections between the topological approaches to recognition in the cases of words and of trees, and model theory.

\paragraph*{Profinite trees and \pdflambda-terms}

One starting point of the connection between regular languages of finite words and trees and the theory of $\lambda$-calculus is given by the Church encoding,
which associates to any ranked alphabet~$\Sigma$ a type~$\Church{\Sigma}$,
and to any tree of $\Sigma$ a closed $\lambda$-term of that type.
For example, the word alphabet with two letters $\Sigma := \{a, b\}$
is encoded as the type
\[
    \Church{\{a, b\}}
    \quad:=\quad
    \underbrace{(\tyo \To \tyo)}_{\text{letter $a$}}
    \ \To \
    \underbrace{(\tyo \To \tyo)}_{\text{letter $b$}}
    \ \To\
    \underbrace{\tyo}_{\text{input}}
    \ \To\
    \underbrace{\tyo}_{\text{output}}
\]
and the word $aaba \in \Sigma^*$ is encoded as the $\lambda$-term
\[
    \lambda(a : \tyo \To \tyo).\,
    \lambda(b : \tyo \To \tyo).\,
    \lambda(c : \tyo).\,
    a\ (b\ (a\ (a\ c)))
    \ .
\]
This demonstrates that ranked alphabets and finite trees can be understood as a fragment of the simply-typed $\lambda$-calculus,
of highly compositional nature.

Inspired by the seminal work of Salvati on higher-order regular languages and denotational semantics~\cite{DBLP:conf/wollic/Salvati09} and the aforementioned topological approach to formal languages,
the profinite $\lambda$-calculus, introduced in~\cite{entics:12280},
extends the simply-typed $\lambda$-calculus with profinite operators coming from automata theory.
In particular, this yields a compactification
\[
    \underbrace{\Tm(A)}_{\text{$\lambda$-terms}}
    \qquad\subseteq\quad
    \underbrace{\ProTm(A)}_{\text{profinite $\lambda$-terms}}
    \qquad\qquad\qquad\text{for any type $A$}
\]
More generally, this approach strives to develop an extension of automata theory to the higher order,
where new letters of the alphabet can be declared on the fly through a $\lambda$-abstraction,
in relation to higher-order model checking~\cite{HOParity} --
see Moreau's PhD thesis for an overview~\cite{moreau:tel-05428993}.

We are therefore in a situation where we have two candidates for the notion of profinite tree, i.e. one based on codensity monads and clones,
and another rooted in the theory of the $\lambda$-calculus and finitary denotational semantics.
In this article, we construct a family of isomorphisms of Stone spaces
\[
    \ProClone{\Tree(\Sigma)}_n
    \quad\cong\quad
    \ProTm(\tyo^n \To \Church{\Sigma})
    \qquad\qquad\qquad\text{for every $n \in \N$}
\]
which assemble together with an isomorphism of clones enriched in Stone spaces.
The proof relies on a parametricity argument, in the sense of Reynolds~\cite{Reynolds83}.
This demonstrates that these two notions coincide,
hence laying a solid foundation for further topological studies of regular tree languages.

\subsection{Plan of the article}

After this introduction,
we start by recalling in \Cref{sec:clones} the notions of clones,
and define four families of clones
\[
    \Tree(\Sigma)
    \qquad,\qquad
    \Endo(Q)
    \qquad,\qquad
    \Act(M)
    \qquad,\qquad
    \Cst(X)
\]
and their relationship in the language of adjoint functors.
In \Cref{sec:functorial-approach},
we revisit these constructions from the viewpoint of Lawvere theories
which provide a high-level view of free/forgetful adjunctions.
In \Cref{sec:profinite-completion-clones},
we recall the notion of codensity monad and its two main examples,
namely the profinite completion of monoids and the ultrafilter monad,
and we define the profinite completion of clones.
We then proceed in \Cref{sec:profinite-completion-isos} to develop the tools necessary to establish that
\begin{align*}
    \ProClone{\Act(M)}
    \quad & \cong\quad
    \Act\big(\ProMon{M}\big)
    \hspace{5cm}\text{ in \Cref{thm:act-profinite-iso}}
    \\
    \ProClone{\Cst(X)}
    \quad & \cong\quad
    \Cst\big(\ProSet X\big)
    \hspace{5cm}\text{in \Cref{thm:cst-profinite-iso}}
\end{align*}
In \Cref{sec:profinite-trees}, we introduce the notion of profinite tree and show in \Cref{prop:free-clones-bidefinable} that they verify a strong form of parametricity.
This permits to arrive smoothly, in \Cref{sec:profinite-lambda-terms},
to an isomorphism of clones enriched in Stone spaces
\[
    \ProClone{\Tree(\Sigma)}
    \quad\cong\quad
    \ProTm\Big(\tyo^{(-)} \To \Church{\Sigma}\Big)
    \hspace{3.2cm}\text{in \Cref{thm:isomorphism-theorem}}
\]
exhibiting profinite trees as a fragment of the profinite $\lambda$-calculus.

\subsection{Related work}

Stone duality has tight links with automata theory, see~\cite{pin2009} and~\cite{Gehrke_vanGool_2024}. Profinite words naturally appear in the Reiterman theorem for pseudovarieties, as proved in~\cite{Reiterman1982} and~\cite{Banaschewski1983}, as sets of profinite equations determine pseudovarieties~\cite{doi:10.1142/2481}. Profinite words have also been used to understand the limitedness problem in~\cite{DBLP:conf/icalp/Torunczyk12} and to show the decidability of weak MSO+U over infinite trees in~\cite{DBLP:conf/stacs/BojanczykT12}. There is a celebrated line of research extending Stone duality to take into account monoid operations on the topological side and residuation operations on the algebraic side, see~\cite{DBLP:conf/icalp/GehrkeGP08, DBLP:conf/icalp/GehrkeGP10}, which inspired the introduction of profinite $\lambda$-terms in~\cite{entics:12280} and thoroughly studied in~\cite{moreau:tel-05428993}.

There is a growing connection between automata theory and $\lambda$-calculus. Salvati has defined the notion of regular language of $\lambda$-terms in~\cite{DBLP:conf/wollic/Salvati09}, using semantic tools. In a more syntactic direction, Hillebrand and Kanellakis established in~\cite{DBLP:conf/lics/HillebrandK96} a link between the regularity of a language and the $\lambda$-definability of its characteristic function. This idea is at the heart of the implicit automata program research by Nguy\~{\^e}n and Pradic started in~\cite{nguyen-pradic-1}, which shows that Hillebrand and Kanellakis' result can be adapted to get a correspondence between star-free languages and a substructural fragment of the simply typed $\lambda$-calculus, see~\cite{titoPhD}. These two directions, semantic and syntactic, yield two notions of languages of $\lambda$-terms of any type which have been shown to coincide in~\cite{moreau_et_al:LIPIcs.CSL.2024.40}.

A lot of different algebras can be considered for trees, see~\cite{DBLP:books/ems/21/Bojanczyk21} for a survey. In~\cite{ESIK2005291Algebraic-recognizability}, Ésik and Weil have shown how preclones, i.e. non-symmetric operads, provide a suitable notion of algebras of trees. In~\cite{DBLP:journals/ijac/EsikW10}, they introduce their block product and use it to give a characterization of first-order definable tree languages.
In the present article,
we use clones as we then obtain an isomorphism theorem, and not only a bijection between the sets of constants as in the case of operads.

The general approach to regular languages and profinite completions using monads began in~\cite{DBLP:conf/dlt/Bojanczyk15} and was further studied in~\cite{Bojaczyk2023} and in~\cite{lmcs:6569}.
Our work directly follows this monadic approach,
given that we consider algebras defined by Lawvere theories
which correspond to finitary monads.
A profinite tree was first used in~\cite{10.1007/978-3-662-43951-7_4},
where it was defined as a Cauchy sequence of profinite trees for an appropriate metric.
The link between metric completions and Stone duality is a fundamental aspect of~\cite{DBLP:conf/icalp/GehrkeGP10} and was elaborated further in relation to Pervin spaces in~\cite{DBLP:conf/RelMiCS/Pin17}.

Codensity is the topic of abundant categorical literature, see e.g. in~\cite{leinstercod, AVERY20161229, DILIBERTI2020106379} where it provides a new point of view on already known monads. It has also been studied in relation to automata theory in a series of papers~\cite{DBLP:conf/fossacs/ChenAMU16, DBLP:conf/mfcs/UrbatACM17, 10.1145/3464691adamek-algebra-monad-profinite}.
We stress the fact that the abstract approach developed in some of these papers
does not directly apply to our setting,
as the theory of clones is infinitely sorted,
see e.g.~\cite[Assumption~2.5]{10.1145/3464691adamek-algebra-monad-profinite}.

Finally, this article uses tools coming from the tradition of higher-order algebraic theories,
in particular the Lawvere theory of clones described in the work of Fiore and Mahmoud on second-order algebraic theories~\cite{10.1007/978-3-642-15155-2_33, DBLP:journals/corr/FioreM14}
which can be understood in the context of higher-order algebraic theories as developed by Arkor and McDermott~\cite{ho-algebraic-theories}, see Arkor's PhD thesis for a revised account~\cite{arkor_2022_phd}.

\subsection{Notations}

\AP
We write $\intro*\Set$ and $\intro*\FinSet$ for the categories of sets and finite sets,
and $\intro*\Mon$ and $\intro*\FinMon$ for the categories of monoids and finite monoids.
We write $\Lan_I F$ for the left Kan extension of $F$ along $I$,
and $\Ran_I F$ for the right Kan extension of $F$ along $I$,
see~\cite[Chapter~X]{MacLane1978}.
A modern account of Stone spaces,
in the context of Priestley duality,
can be found in~\cite{Gehrke_vanGool_2024}.

\section{Abstract clones, tree automata, and monoids and sets}
\label{sec:clones}

In this section,
we recall the notion of "abstract clones",
introduced by Hall~\cite{cohn1981universal}.
We then go on to show that every "ranked alphabet"~$\Sigma$ induces a "clone"~$\Tree(\Sigma)$,
whose operations are ranked trees,
and which are freely generated by the elements of $\Sigma$.
Clones are our algebras to treat regular tree languages.

\begin{definition}
    \AP\label{def:clone}
    A ""clone""~$C$ is a family of sets
    \[
        C_n
        \qquad
        \text{where $n$ ranges over $\N$}
    \]
    together, for every $n \in \N$, with elements called ""variables""
    \[
        v_1, \dots, v_n
        \ \in\
        C_n
    \]
    and, for every $m, n \in \N$, with functions called ""substitutions""
    \[
        (-_1)[-_2]
        \quad:\quad
        C_n \times (C_m)^n
        \ \longto\
        C_m
    \]
    such that the three following axioms are verified:
    \begin{align*}
        v_k[x_1, \dots, x_n]
        \  & =\
        x_k
        \\
        x[v_1, \dots, v_n]
        \  & =\
        x
        \\
        x[y_1[\vec{z}], \dots, y_n[\vec{z}]]
        \  & =\
        x[\vec{y}][\vec{z}]
    \end{align*}
    A ""clone morphism""~$\varphi : C \to D$ is a family of functions
    \[
        \varphi_n
        \quad:\quad
        C_n
        \ \longto\
        D_n
        \qquad\text{where $n$ ranges over $\N$}
    \]
    which preserves variables and substitution functions, i.e.
    \begin{align*}
        \varphi_n(v_k)
        \  & =\
        v_k
        \\
        \varphi_m(x[\vec{y}])
        \  & =\
        \varphi_n(x)[\varphi_m(y_1), \dots, \varphi_m(y_n)]
    \end{align*}
    Clones and their morphisms form a category~$\intro*\Clone$
\end{definition}

\begin{remark}
    \AP\label{rmk:enriched-clones}
    We stress the fact that the above definition can be generalized to any cartesian category~$\V$,
    as described in the work of Fiore and Mahmoud~\cite[\S7.1]{DBLP:journals/corr/FioreM14}.
    Each $C_n$ is then an object of $\V$,
    and the variables and substitution are morphisms of $\V$
    \begin{align*}
        v_k
        \quad & :\quad
        1_{\V}
        \ \longto\
        C_n
        \qquad\text{for all $1 \le k \le n$}
        \\
        (-_1)[-_2]
        \quad & :\quad
        C_n \times (C_m)^n
        \ \longto\
        C_m
    \end{align*}
    where $1_{\V}$ is the terminal object of $\V$, and $\times$ its cartesian product.
    The equations of \Cref{def:clone} can then be rephrased in terms of the commutativity of the three diagrams
    \[
        \begin{tikzcd}[ampersand replacement=\&]
            {(C_m)^n} \&\& {C_n \times (C_m)^n} \\
            \&\& {C_m}
            \arrow["{v_k \times \Id}", from=1-1, to=1-3]
            \arrow["{\pi_k}"', from=1-1, to=2-3]
            \arrow["{(-_1)[-_2]}", from=1-3, to=2-3]
        \end{tikzcd}
        \qquad
        \begin{tikzcd}[ampersand replacement=\&]
            {C_n} \&\& {C_n \times (C_n)^n} \\
            \&\& {C_n}
            \arrow["{\Id \times v_1 \times \dots \times v_n}", from=1-1, to=1-3]
            \arrow["\Id"', from=1-1, to=2-3]
            \arrow["{(-_1)[-_2]}", from=1-3, to=2-3]
        \end{tikzcd}
    \]
    \[
        \begin{tikzcd}[ampersand replacement=\&]
            {C_n \times (C_m)^n \times (C_l)^m} \&\&\&\& {C_m \times (C_l)^m} \\
            {C_n \times (C_m \times (C_l)^m)^n} \&\&\& {C_n \times (C_l)^n} \& {C_l}
            \arrow["{(-_1)[-_2] \times \Id}", from=1-1, to=1-5]
            \arrow[from=1-1, to=2-1]
            \arrow["{(-_1)[-_2]}", from=1-5, to=2-5]
            \arrow["{\Id \times ((-_1)[-_2])^n}"', from=2-1, to=2-4]
            \arrow["{(-_1)[-_2]}"', from=2-4, to=2-5]
        \end{tikzcd}
    \]
    This yields a category of $\V$-""enriched clones"".
\end{remark}

In the rest of this section,
we define the "clones"~$\Tree(\Sigma)$ and~$\Endo(Q)$,
which emphasize the relation with deterministic, bottom-up tree automata,
and we show how to encode monoids~$M$ and sets~$X$ as the "clones"~$\Act(M)$ and~$\Cst(X)$.

\begin{definition}
    \AP\label{def:endo}
    For any set~$Q$, the clone~$\Endo(Q)$ is defined as the one with the sets
    \begin{align*}
        \intro*\Endo(Q)
        \quad & :=\quad
        \Set(Q^n, Q)
        \qquad\text{for any $n \in \N$}
        \\
              & \ =\quad
        \{\text{set-theoretic functions $Q^n \to Q$}\}
    \end{align*}
    equipped with "variables"~$v_k \in \Endo(Q)_n$ defined as
    \[
        v_k
        \quad:\quad
        (q_1, \dots, q_n)
        \ \longmapsto\
        q_k
    \]
    and with "substitution" function
    \[
        f[g_1, \dots, g_n]
        \quad:\quad
        \vec{q}
        \ \longmapsto\
        f(g_1(\vec{q}), \dots, g_n(\vec{q}))
    \]
\end{definition}

\begin{definition}
    \AP\label{def:tree}
    A ""ranked alphabet""~$\Sigma$ is a finite set of letters with finite arities
    \[
        \Sigma
        \quad:=\quad
        \{a_1 : n_1, \dots, a_l : n_l\}
    \]
    The clone $\Tree(\Sigma)$ is defined as the one with sets
    \[
        \intro*\Tree(\Sigma)_n
        \quad:=\quad
        \{\text{ranked trees on $\Sigma$ with shared variables in $v_1, \dots, v_n$}\}
    \]
    and with variables the leaf trees, and substitution given by tree grafting.
\end{definition}

\begin{remark}
    \label{rmk:endo-automata}
    The "clone"~$\Tree(\Sigma)$ has a universal property:
    it is the free clone on the generators in the "ranked alphabet"~$\Sigma$.
    This means that if $\Sigma := \{a_1 : n_1, \dots, a_l : n_l\}$,
    then for every "clone"~$C$ there is a bijection between the sets
    \[
        \Clone(\Tree(\Sigma), C)
        \quad\cong\quad
        C_{n_1} \times \dots \times C_{n_l}
    \]
    which is natural in the "clone"~$C$.
    In particular, when $C$ is a clone of the form $\Endo(Q)$ for a set~$Q$,
    we obtain that clone morphisms
    \[
        \delta
        \quad:\quad
        \Tree(\Sigma)
        \ \longto\
        \Endo(Q)
    \]
    are in bijection with tuples of transition functions
    \[
        \delta_{a_1} \ :\ Q^{n_1} \to Q
        \quad,\quad\dots\quad,\quad
        \delta_{a_l} \ :\ Q^{n_l} \to Q
    \]
    which amounts to the structure of a
    bottom-up, deterministic automaton without final states
    with $Q$ as set of states.
\end{remark}

\begin{definition}
    \AP\label{def:act}
    For any monoid~$M$, the clone~$\Act(M)$ is defined as the family of sets
    \[
        \intro*\Act(M)_n
        \quad:=\quad
        M \times \{x_1, \dots, x_n\}
        \qquad\text{for any $n \in \N$}
    \]
    equipped with "variables"~$v_k \in \Act(M)_n$ defined as
    \[
        v_k
        \quad:=\quad
        (1_M, x_k)
    \]
    and with "substitutions"
    \[
        (a, x_l)
        \left[
            (b_1, x_{k_1}), \dots, (b_n, x_{k_m})
            \right]
        \quad:=\quad
        (a \cdot b_l\ ,\ x_{k_l})
    \]
\end{definition}

\AP
The assignment of \Cref{def:act} forms a functor $\Act : \Mon \to \Clone$
that is fully faithful,
hence allowing us to think of monoids as certain clones.
We now describe how it appears in an adjunction.

\begin{proposition}
    \label{prop:act-adjunction}
    We have an adjunction
    \[
        \begin{tikzcd}[ampersand replacement=\&]
            \Mon \&\& \Clone
            \arrow[""{name=0, anchor=center, inner sep=0}, "\Act", shift left=2, from=1-1, to=1-3]
            \arrow[""{name=1, anchor=center, inner sep=0}, "{C_1\,\mapsfrom\,C}", shift left=2, from=1-3, to=1-1]
            \arrow["\dashv"{anchor=center, rotate=-90}, draw=none, from=0, to=1]
        \end{tikzcd}
    \]
    which is coreflective, i.e. the left adjoint~$\Act$ is fully faithful.
\end{proposition}

\begin{proof}
    Let $M$ be a monoid, $C$ be a "clone".
    We show that every monoid homomorphism~$\varphi : M \to C_1$ lifts in a unique way as a clone morphism~$\psi : \Act(M) \to C$ such that the following diagram commutes
    \[
        \begin{tikzcd}[ampersand replacement=\&]
            {\Act(M)_1} \\
            M \& {C_1}
            \arrow["{\psi_1}", from=1-1, to=2-2]
            \arrow["\cong"{marking, allow upside down}, draw=none, from=2-1, to=1-1]
            \arrow["\varphi"', from=2-1, to=2-2]
        \end{tikzcd}
    \]
    For every $n \in \N$ and every $(a, v_k) \in \Act(M)_n$, we have
    \[
        (a, x_k)
        \ =\
        \underbrace{(a, x_1)}_{\text{in } \Act(M)_1}[(1_M, x_k)]
    \]
    so for any such clone morphism~$\psi$, we have that
    \[
        \psi_n(a, x_k)
        \ =\
        \psi_1(a, x_1)[\psi_n(1_M, x_k)]
        \ =\
        \varphi(a)[v_k]
    \]
    so $\psi$ is determined by $\varphi$.
    Conversely, taking this equation as a definition of $\psi_n$ for every $n \in \N$ yields the unique extension of $\varphi$ as a clone morphism~$\psi : \Act(M) \to C$.
    The coreflectiveness amounts to the fact that the unit is an isomorphism.
\end{proof}

\begin{definition}
    \AP\label{def:cst}
    For any set~$X$, the clone~$\Cst(X)$ is defined as the family of sets
    \[
        \intro*\Cst(X)_n
        \quad:=\quad
        X \sqcup \{v_1, \dots, v_n\}
        \qquad\text{for any $n \in \N$}
    \]
    together with "variables"~$v_1, \dots, v_n$ and with "substitutions"
    \begin{align*}
        x[c_1, \dots c_n]
        \quad & :=\quad
        x
        \qquad\text{for $x \in X$}
        \\
        v_k[c_1, \dots, c_n]
        \quad & :=\quad
        c_k
        \qquad\text{for $1 \le k \le n$}
    \end{align*}
\end{definition}

In the same way as in \Cref{prop:act-adjunction},
we have the following adjunction.

\begin{proposition}
    \label{prop:cst-adjunction}
    We have an adjunction
    \[
        \begin{tikzcd}[ampersand replacement=\&]
            \Set \&\& \Clone
            \arrow[""{name=0, anchor=center, inner sep=0}, "\Cst", shift left=2, from=1-1, to=1-3]
            \arrow[""{name=1, anchor=center, inner sep=0}, "{C_0\,\mapsfrom\,C}", shift left=2, from=1-3, to=1-1]
            \arrow["\dashv"{anchor=center, rotate=-90}, draw=none, from=0, to=1]
        \end{tikzcd}
    \]
    which is coreflective, i.e. the left adjoint~$\Cst$ is fully faithful.
\end{proposition}

\begin{proof}
    Let $X$ be a set, $C$ be a "clone". We show that every set-theoretic function~$f : X \to C_0$ lifts in a unique way as a clone morphism~$\psi : \Cst(X) \to C$ such that the following diagram commutes
    \[
        \begin{tikzcd}[ampersand replacement=\&]
            {\Cst(X)_0} \\
            X \& {C_0}
            \arrow["{\psi_0}", from=1-1, to=2-2]
            \arrow["\cong"{marking, allow upside down}, draw=none, from=2-1, to=1-1]
            \arrow["f"', from=2-1, to=2-2]
        \end{tikzcd}
    \]
    For every $n \in \N$ and every $x \in X$, we have
    \[
        \underbrace{x}_{\text{in } \Cst(X)_n}
        \ =\
        \underbrace{x}_{\text{in } \Cst(X)_0}[\,]
    \]
    so for any such clone morphism~$\psi$, we obtain that
    \begin{align*}
        \psi_n(x)
        \  & =\
        \psi_0(x)[\,]
        \ =\
        f(x)[\,]
        \\
        \psi_n(v_k)
        \  & =\
        v_k
        \hspace{4cm}\text{for every $1 \le k \le n$}
    \end{align*}
    so $\psi$ is determined by $f$.
    Taking these two equations as a definition for $\psi$ yields the unique extension of $f$ as a clone morphism~$\psi : \Cst(X) \to C$.
    The coreflectiveness amounts to the fact that the unit is an isomorphism.
\end{proof}

\section{A functorial viewpoint on clones and related structures}
\label{sec:functorial-approach}

In this section,
we describe a categorical approach on "clones" based on the notion of Lawvere theory,
which induces naturally the two fully faithful functors
\[
    \Act
    \quad:\quad
    \Mon
    \ \longto\
    \Clone
    \qquad\text{and}\qquad
    \Cst
    \quad:\quad
    \Set
    \ \longto\
    \Clone
\]
which were described in \Cref{sec:clones},
rooted in universal algebra and functorial semantics.
We stress the fact that this approach is directly related to the monadic approach developed by \Bojanczyk~\cite{DBLP:conf/dlt/Bojanczyk15} and then later refined by \Bojanczyk, Klein and Salamanca~\cite{Bojaczyk2023}
as our "theories" yield finitary monads on $\Set$ and $\Set^\N$.
This perspective will turn out to be crucial in the further developments in \Cref{sec:profinite-completion-isos}.
For that matter, we now make precise the notion of theory we consider in this paper.

\begin{definition}
    \AP\label{def:finite-product-theory}
    A ""theory"" is a small category~$\T$ with finite cartesian products.

    A ""model""~$X$ of the "theory"~$\T$ is a product-preserving functor~$X : \T \to \Set$,
    and a morphism of models is a natural transformation.
    These form a category, that we write $\intro*\Mod(\T)$.
\end{definition}

\begin{remark}
    The functorial approach to universal algebra,
    as pioneered by Lawvere~\cite{LawvereThesis},
    has a slightly different notion of "theory" which makes the notion of arity explicit.
    Indeed, a Lawvere theory is a category~$\T$ with finite cartesian products, equipped with a functor
    \[
        \op{\FinSet}
        \ \longto\
        \T
    \]
    which is identity-on-objects and which preserves finite cartesian products.

    In this article, we do not use the notion of arity
    so our notion of "theory" described in \Cref{def:finite-product-theory}
    is simply a small category with finite products.
    This is a standard notion,
    considered for example in the book of Adámek, Rosický, and Vitale~\cite[Definition~1.1]{Admek2010}.
\end{remark}

We now introduce the three central "theories" of this article,
$\TMon$, $\TSet$, and $\TClone$,
whose "models" are monoids, sets, and "clones" respectively.

\begin{definition}
    \AP\label{def:theory-sets}
    We write $\intro*\TMon$ for the cartesian category
    whose objects are natural numbers
    with cartesian product given by addition,
    and with hom-sets
    \begin{align*}
        \TMon(n, 1)
        \quad & :=\quad
        \{a_1, \dots, a_n\}^*
        \\
              & \ =\quad
        \{\text{words with letters among $a_1$, \dots, $a_n$}\}
    \end{align*}
    The "models" of this "theory" are monoids,
    i.e. there is an equivalence of categories
    \[
        \Mon
        \quad\cong\quad
        \Mod(\TMon)
    \]
\end{definition}

\begin{definition}
    \AP\label{def:theory-monoids}
    We write $\intro*\TSet$ for the cartesian category
    whose objects are natural numbers
    with cartesian product given by addition,
    and with hom-sets
    \[
        \TSet(n, 1)
        \quad:=\quad
        \{a_1, \dots, a_n\}
    \]
    The "models" of this "theory" are sets,
    i.e. there is an equivalence of categories
    \[
        \Set
        \quad\cong\quad
        \Mod(\TSet)
    \]
\end{definition}

\begin{definition}
    \AP\label{def:theory-clones}
    We write $\intro*\TClone$ for the cartesian category
    whose objects are "ranked alphabets"
    with cartesian product given by concatenation,
    and with hom-sets
    \begin{align*}
        \TClone(\Sigma, \{b : n\})
        \quad & :=\quad
        \Tree(\Sigma)_n
        \\
              & \ =\quad
        \{
        \text{ranked trees on $\Sigma$ with shared variables in $v_1$, \dots, $v_n$}\}
    \end{align*}
    This category is the "theory" of "clones",
    i.e. there is an equivalence of categories
    \[
        \Clone
        \quad\cong\quad
        \Mod(\TClone)
    \]
\end{definition}

\begin{remark}
    \label{rmk:clone-tclone}
    The category~$\TClone$ has been introduced by Fiore and Mahmoud
    in their seminal study of second-order algebraic theories~\cite{10.1007/978-3-642-15155-2_33, DBLP:journals/corr/FioreM14},
    where it is written $\mathbf{M}$.
    The fact that this category is the "theory" of clones is a corollary of their work,
    see~\cite[Theorem~7.6]{DBLP:journals/corr/FioreM14}.
    This category has then been identified as the order $2$ fragment of the simply-typed $\lambda$-calculus by Arkor and McDermott in their work on higher-order algebraic theories~\cite{arkor_2022_phd}.

    For the reader's convenience,
    we outline here the equivalence,
    described in the work of Fiore and Mahmoud,
    between the two categories~$\Clone$ and $\Mod(\TClone)$.
    \begin{itemize}
        \item To any "clone"~$C$,
              we associate the "model"~$X : \TClone \to \Set$
              that sends any "ranked alphabet"~$\Sigma$ defined as $\{a_1 : n_1, \dots, a_l : n_l\}$ on the set
              \[
                  X(\Sigma)
                  \quad:=\quad
                  C_{n_1} \times \dots \times C_{n_l}
              \]
              and any tree in the hom-set~$\TClone(\Sigma, \{b : n\})$ on the set-theoretic function
              \[
                  C_{n_1} \times \dots \times C_{n_l}
                  \ \longto\
                  C_{n}
              \]
              computed inductively on the tree using the "variables" and "substitution" of $C$.
        \item To any "model"~$X : \TClone \to \Set$,
              we associate the "clone"~$C$ whose sets are defined as
              \[
                  C_{n}
                  \quad:=\quad
                  X(\{b : n\})
                  \qquad\text{for any $n \in \N$}
              \]
              equipped with the "variables" defined as the image by $X$ of the leaf trees
              \[
                  v_k
                  \quad\in\quad
                  \TClone\big(\varnothing, \{b : n\}\big)
              \]
              and with the "substitutions" defined as the image by $X$ of the tree
              \[
                  \begin{tikzcd}[ampersand replacement=\&, row sep=tiny, column sep=tiny]
                      \&\&\& {a_0} \\
                      \& {a_1} \&\& \dots \&\& {a_n} \\
                      {v_1} \& \dots \& {v_m} \&\& {v_1} \& \dots \& {v_m}
                      \arrow[no head, from=1-4, to=2-2]
                      \arrow[no head, from=1-4, to=2-4]
                      \arrow[no head, from=1-4, to=2-6]
                      \arrow[no head, from=2-2, to=3-1]
                      \arrow[no head, from=2-2, to=3-2]
                      \arrow[no head, from=2-2, to=3-3]
                      \arrow[no head, from=2-6, to=3-5]
                      \arrow[no head, from=2-6, to=3-6]
                      \arrow[no head, from=2-6, to=3-7]
                  \end{tikzcd}
                  \in\quad
                  \TClone\big(\{a_0 : n, a_1 : m, \dots, a_n : m\}, \{b : n\}\big)
              \]
    \end{itemize}
    It is worth noting that for any "ranked alphabet"~$\Sigma$,
    seen as an object of $\TClone$,
    this equivalence between the categories~$\Clone$ and~$\Mod(\TClone)$
    relates free "clones"~$\Tree(\Sigma)$ from \Cref{def:tree} with representable "models"
    \[
        \TClone(\Sigma, -)
        \quad:\quad
        \TClone
        \ \longto\
        \Set
    \]
\end{remark}

We now proceed to relate the three "theories" that we have introduced so far.
At the heart of universal algebra lies the construction of free structures,
which can be built syntactically by inductively adding operations and quotiented by the required equations.
This construction of free structures is adjoint to the operation of forgetting parts of the structure.
One of the benefits of the notion of "theory" considered in this article
is that it puts this free/forgetful adjunction,
central to the field of universal algebra,
in the broader context of category theory.
Indeed, we then have the following dictionary:
\[
    \begin{tikzcd}[ampersand replacement=\&]
        {\text{free construction}} \&\& {\text{left Kan extension}} \\
        {\text{forgetful functor}} \&\& {\text{precomposition}}
        \arrow["{\text{seen as}}", between={0.1}{0.9}, dotted, from=1-1, to=1-3]
        \arrow["{\text{adjoints}}"', dashed, no head, from=1-1, to=2-1]
        \arrow["{\text{adjoints}}", dashed, no head, from=1-3, to=2-3]
        \arrow["{\text{seen as}}"', between={0.1}{0.9}, dotted, from=2-1, to=2-3]
    \end{tikzcd}
\]
This leads us to the following fact,
which is a special case of a result by Lawvere~\cite{LawvereThesis}.

\begin{proposition}
    \AP\label{prop:pull-push-theories}
    Any cartesian product-preserving functor
    \[
        F
        \quad:\quad
        \T
        \ \longto\
        \T'
    \]
    between "theories" induces an adjunction
    \[
        \begin{tikzcd}[ampersand replacement=\&]
            {\Mod(\T)} \&\& {\Mod(\T')}
            \arrow[""{name=0, anchor=center, inner sep=0}, "{\Lan_F}", shift left=2, from=1-1, to=1-3]
            \arrow[""{name=1, anchor=center, inner sep=0}, "{(-) \circ F}", shift left=2, from=1-3, to=1-1]
            \arrow["\dashv"{anchor=center, rotate=-90}, draw=none, from=0, to=1]
        \end{tikzcd}
    \]
    given that the left Kan extension of a product-preserving functor is product-preserving.

    Moreover, if the functor~$F$ is fully faithful,
    then so is the left adjoint~$\Lan_{F}$.
\end{proposition}

As a first illustration,
we now recall how to get back the usual free/forgetful adjunction between the categories of sets and monoids.

\begin{example}
    \label{ex:adj-set-mon}
    We consider the cartesian product-preserving functor
    \[
        F
        \quad:\quad
        \TSet
        \ \longto\
        \TMon
    \]
    which is the identity on objects,
    and whose action on morphisms is the inclusion
    \[
        \{a_1, \dots, a_n\}
        \quad\subseteq\quad
        \{a_1, \dots, a_n\}^*
    \]
    of letters into words.
    Under the identifications provided by the equivalences of categories
    \[
        \Set
        \quad\cong\quad
        \Mod(\TSet)
        \qquad\text{and}\qquad
        \Mon
        \quad\cong\quad
        \Mod(\TMon)
    \]
    the associated adjunction,
    obtained from \Cref{prop:pull-push-theories},
    is the usual free/forgetful adjunction
    \[
        \begin{tikzcd}[ampersand replacement=\&]
            \Set \&\&\& \Mon
            \arrow[""{name=0, anchor=center, inner sep=0}, "{(-)^*}", shift left=2, from=1-1, to=1-4]
            \arrow[""{name=1, anchor=center, inner sep=0}, "{\text{underlying set}}", shift left=2, from=1-4, to=1-1]
            \arrow["\dashv"{anchor=center, rotate=-90}, draw=none, from=0, to=1]
        \end{tikzcd}
    \]
\end{example}

We now proceed to revisit the adjunctions from
\Cref{prop:act-adjunction} and \Cref{prop:cst-adjunction},
whose left adjoints are the fully faithful functors $\Act$ and $\Cst$ encoding monoids and sets as "clones",
from the categorical perspective of "theories".

\begin{example}
    \label{ex:adj-mon-clone}
    We define a cartesian product-preserving functor
    \[
        F
        \quad:\quad
        \TMon
        \ \longto\
        \TClone
    \]
    The functor~$F$ sends an object~$n$ of $\TMon$,
    i.e. a natural number,
    on the "ranked alphabet"
    \[
        F(n)
        \quad:=\quad
        \{a_1 : 1, \dots, a_n : 1\}
    \]
    consisting of $n$ unary generators.
    The functor~$F$ sends a morphism~$w \in \TMon(n, 1)$,
    i.e. a word~$w := a_{i_1} \dots a_{i_l} \in \{a_1, \dots, a_n\}^*$,
    on the filiform tree
    \[
        F(w)
        \quad:=\quad
        \begin{tikzcd}[ampersand replacement=\&, row sep=tiny, column sep=tiny]
            {a_{i_1}} \\
            \vdots \\
            {a_{i_l}} \\
            {v_1}
            \arrow[no head, from=1-1, to=2-1]
            \arrow[no head, from=2-1, to=3-1]
            \arrow[no head, from=3-1, to=4-1]
        \end{tikzcd}
        \qquad\in\qquad
        \Tree(\{a_1 : 1, \dots, a_n : 1\})_1
    \]
    Using the correspondence between "clones" and "models" of $\TClone$ of  \Cref{rmk:clone-tclone},
    we compute that precomposition by $F$ sends a "clone"~$C$
    on the monoid~$C_1$ of its unary operations.
    As a consequence,
    the adjunction associated by \Cref{prop:pull-push-theories} is the adjunction from \Cref{prop:act-adjunction}
    \[
        \begin{tikzcd}[ampersand replacement=\&]
            \Mon \&\& \Clone
            \arrow[""{name=0, anchor=center, inner sep=0}, "\Act", shift left=2, from=1-1, to=1-3]
            \arrow[""{name=1, anchor=center, inner sep=0}, "{C_1\,\mapsfrom\,C}", shift left=2, from=1-3, to=1-1]
            \arrow["\dashv"{anchor=center, rotate=-90}, draw=none, from=0, to=1]
        \end{tikzcd}
    \]
    The fully faithfulness of the functor~$F$ implies the one of $\Act$.
\end{example}

\begin{example}
    \label{ex:adj-set-clone}
    We define a cartesian product-preserving functor
    \[
        F
        \quad:\quad
        \TMon
        \ \longto\
        \TClone
    \]
    The functor~$F$ sends an object~$n$ of $\TSet$,
    i.e. a natural number,
    on the "ranked alphabet"
    \[
        F(n)
        \quad:=\quad
        \{a_1 : 0, \dots, a_n : 0\}
    \]
    consisting of $n$ unary generators.
    The functor~$F$ sends a morphism~$w \in \Set(n, 1)$,
    i.e. an element~$a_k$ for $1 \le k \le n$,
    on the leaf tree
    \[
        F(a_k)
        \quad:=\quad
        a_k
        \qquad\in\qquad
        \Tree(\{a_1 : 0, \dots, a_n : 0\})_0
    \]
    Again, the correspondence of \Cref{rmk:clone-tclone}
    shows that precomposition by $F$ sends a "clone"~$C$
    on the set~$C_0$ of its constants.
    As a consequence,
    the adjunction associated by \Cref{prop:pull-push-theories} is the adjunction from \Cref{prop:act-adjunction}
    \[
        \begin{tikzcd}[ampersand replacement=\&]
            \Set \&\& \Clone
            \arrow[""{name=0, anchor=center, inner sep=0}, "\Cst", shift left=2, from=1-1, to=1-3]
            \arrow[""{name=1, anchor=center, inner sep=0}, "{C_0\,\mapsfrom\,C}", shift left=2, from=1-3, to=1-1]
            \arrow["\dashv"{anchor=center, rotate=-90}, draw=none, from=0, to=1]
        \end{tikzcd}
    \]
    The fully faithfulness of the functor~$F$ implies the one of $\Cst$.
\end{example}

\section{Profinite completions of clones}
\label{sec:profinite-completion-clones}
In this section,
we recall the general notion of "codensity monad"
and its relation with ultrafilters and profinite words,
before introducing the "profinite completion of clones"
based on the notion of "locally finite clone".

\begin{definition}
    \AP\label{def:codensity-monad}
    For any functor~$I : \C_f \to \C$ between categories,
    its ""codensity monad""
    \[
        \intro*\cody{I}
        \quad:\quad
        \C
        \ \longto\
        \C
    \]
    is the right Kan extension of the functor~$I$ against itself,
    i.e.
    \[
        \begin{tikzcd}[ampersand replacement=\&]
            {\C_f} \& \C \& \C
            \arrow["I", from=1-1, to=1-2]
            \arrow[""{name=0, anchor=center, inner sep=0}, "I"', curve={height=30pt}, from=1-1, to=1-3]
            \arrow["{\cody{I}}", dashed, from=1-2, to=1-3]
            \arrow["\epsilon", between={0.2}{0.8}, Rightarrow, from=1-2, to=0]
        \end{tikzcd}
    \]
\end{definition}

The monad structure on the endofunctor~$\cody{I}$
can be built using the universal property of the right Kan extension:
the unit~$\eta : \Id \To \cody{I}$ is defined as the factorization
\[
    \begin{tikzcd}[ampersand replacement=\&]
        {\C_f} \& \C \& \C
        \arrow["I", from=1-1, to=1-2]
        \arrow[""{name=0, anchor=center, inner sep=0}, "I"', curve={height=30pt}, from=1-1, to=1-3]
        \arrow["\Id", from=1-2, to=1-3]
        \arrow[between={0.2}{0.8}, equals, from=1-2, to=0]
    \end{tikzcd}
    \qquad=\qquad
    \begin{tikzcd}[ampersand replacement=\&]
        {\C_f} \& \C \& \C
        \arrow["I", from=1-1, to=1-2]
        \arrow[""{name=0, anchor=center, inner sep=0}, "I"', curve={height=30pt}, from=1-1, to=1-3]
        \arrow[""{name=1, anchor=center, inner sep=0}, "{\cody{I}}"', from=1-2, to=1-3]
        \arrow[""{name=2, anchor=center, inner sep=0}, "\Id", curve={height=-24pt}, from=1-2, to=1-3]
        \arrow["\epsilon", between={0.2}{0.8}, Rightarrow, from=1-2, to=0]
        \arrow["\eta", between={0.2}{0.8}, Rightarrow, from=2, to=1]
    \end{tikzcd}
\]
and the multiplication~$\mu : \cody{I} \circ \cody{I} \To \cody{I}$ is defined as the factorization
\[
    \begin{tikzcd}[ampersand replacement=\&]
        {\C_f} \& \C \& \C \& \C
        \arrow["I", from=1-1, to=1-2]
        \arrow[""{name=0, anchor=center, inner sep=0}, curve={height=30pt}, from=1-1, to=1-3]
        \arrow[""{name=1, anchor=center, inner sep=0}, "I"', curve={height=50pt}, from=1-1, to=1-4]
        \arrow[curve={height=-30pt}, draw=none, from=1-1, to=1-4]
        \arrow[""{name=2, anchor=center, inner sep=0}, "{\cody{I}}", from=1-2, to=1-3]
        \arrow["{\cody{I}}", from=1-3, to=1-4]
        \arrow["\epsilon", between={0.2}{0.8}, Rightarrow, from=1-2, to=0]
        \arrow["\epsilon", between={0.6}{0.8}, Rightarrow, from=2, to=1]
    \end{tikzcd}
    \qquad=\qquad
    \begin{tikzcd}[ampersand replacement=\&]
        {\C_f} \& \C \&\& \C
        \arrow["I", from=1-1, to=1-2]
        \arrow[""{name=0, anchor=center, inner sep=0}, "I"', curve={height=30pt}, from=1-1, to=1-4]
        \arrow[""{name=1, anchor=center, inner sep=0}, draw=none, from=1-1, to=1-4]
        \arrow[""{name=2, anchor=center, inner sep=0}, "{\cody{I}}"', from=1-2, to=1-4]
        \arrow[""{name=3, anchor=center, inner sep=0}, "{\cody{I} \circ \cody{I}}", curve={height=-24pt}, from=1-2, to=1-4]
        \arrow["\epsilon", between={0.2}{0.8}, Rightarrow, from=1, to=0]
        \arrow["\mu", between={0.2}{0.8}, Rightarrow, from=3, to=2]
    \end{tikzcd}
\]

We stress the fact that,
in \Cref{def:codensity-monad},
$\C_f$ and $\C$ are any categories and $I : \C_f \to \C$ is any functor.
The notations suggest that $\C_f$ is the full subcategory of $\C$ consisting of finite objects and that $I$ is the associated inclusion,
as this will be the case in practice in \Cref{ex:ultrafilters} and \Cref{ex:profinite-mon},
but there is no such requirement when defining the general notion of "codensity monad".

A standard fact about right Kan extensions is that under mild hypothesis,
they can be computed by limits, see e.g.~\cite[Chapter X, \S3, Theorem~1]{MacLane1978},~\cite[Theorem~6.2.1]{riehl2017}.
This observation amounts to a first step towards profinite topological spaces,
given that these are limits of finite, discrete spaces.

\begin{proposition}
    \label{prop:codensity-limit}
    For any functor~$I : \C_f \to \C$,
    if $\C_f$ is an essentially small category
    and $\C$ is complete,
    then the "codensity monad"~$\cody{I}$ exists,
    and is computed as the limit
    \[
        \cody{I}(X)
        \quad\cong\quad
        \lim_{\substack{X' \in \C_f \\ X \to I X'}} I X'
        \qquad\text{for every object~$X$ of $\C$}
    \]
\end{proposition}

We now give two examples of well-known "codensity monads",
whose further investigation is the topic of \Cref{sec:profinite-completion-isos}.
Other examples of "codensity monads",
in relation to probability theory,
continuations, Isbell duality, and many other topics
can be found in~\cite{leinstercod, DILIBERTI2020106379, https://doi.org/10.48550/arxiv.2509.26197}.

\begin{example}
    \label{ex:ultrafilters}
    A classical result by Kennison and Gildenhuy~\cite{KENNISON1971317}
    states that the ultrafilter monad~$\ProSet : \Set \to \Set$
    is the "codensity monad" of the usual inclusion~$\FinSet \to \Set$.
    \[
        \begin{tikzcd}[ampersand replacement=\&]
            \FinSet \& \Set \& \Set
            \arrow[from=1-1, to=1-2]
            \arrow[""{name=0, anchor=center, inner sep=0}, curve={height=30pt}, from=1-1, to=1-3]
            \arrow[""{name=1, anchor=center, inner sep=0}, draw=none, from=1-1, to=1-3]
            \arrow[""{name=1p, anchor=center, inner sep=0}, phantom, from=1-1, to=1-3, start anchor=center, end anchor=center]
            \arrow["{\intro*\ProSet}", dashed, from=1-2, to=1-3]
            \arrow[between={0.3}{0.8}, Rightarrow, from=1p, to=0]
        \end{tikzcd}
    \]
\end{example}

\begin{example}
    \label{ex:profinite-mon}
    The inclusion~$\FinMon \to \Mon$ of finite monoids induces a "codensity monad"
    \[
        \begin{tikzcd}[ampersand replacement=\&]
            \FinMon \& \Mon \& \Mon
            \arrow[from=1-1, to=1-2]
            \arrow[""{name=0, anchor=center, inner sep=0}, curve={height=30pt}, from=1-1, to=1-3]
            \arrow[""{name=1, anchor=center, inner sep=0}, draw=none, from=1-1, to=1-3]
            \arrow[""{name=1p, anchor=center, inner sep=0}, phantom, from=1-1, to=1-3, start anchor=center, end anchor=center]
            \arrow["{\intro*\ProMon{(-)}}", dashed, from=1-2, to=1-3]
            \arrow[between={0.3}{0.8}, Rightarrow, from=1p, to=0]
        \end{tikzcd}
    \]
    called the profinite completion of monoids.
    When applied to the free monoid on a finite set~$\Sigma^*$,
    it yields the monoid of profinite words~$\ProMon{\Sigma^*}$.
\end{example}

These two examples show that the ultrafilter monad and the profinite completion of monoids arise as the "codensity monads" associated to appropriate notions of finite structures,
i.e. finite sets and finite monoids respectively.

We now return to "clones",
and we introduce the notion of finiteness that we consider.

\begin{definition}
    \AP\label{def:locally-finite-clone}
    A "clone"~$C$ is ""locally finite""
    when the sets~$C_n$ are finite for every~$n \in \N$.
    We write $\intro*\FinClone$ for the full subcategory of $\Clone$ whose objects are "locally finite clones".
\end{definition}

\begin{remark}
    \label{rmk:infinite-trees}
    We stress the fact that, in the adjacent case of infinite trees,
    a subtle notion of finite structure
    -- requiring an hypothesis of finite generation --
    must be considered in order to establish a relation with regular languages,
    see e.g.~\cite[\S3]{lmcs:4747}.
    Yet, this article treats the case of finite trees, not infinite ones.
    Moreover, we work in a setting where
    automata structures (without final states) associated to a "ranked alphabet"~$\Sigma$ on a set~$Q$
    correspond to morphisms of clones
    \[
        \Tree(\Sigma)
        \ \longto\
        \Endo(Q)
    \]
    It is therefore natural that
    the "clone"~$\Endo(Q)$ verifies our chosen finiteness criterion,
    i.e. it is a "locally finite clone",
    if and only if the set~$Q$ is finite.
\end{remark}

Given that the category $\FinClone$ is essentially small,
and that the category $\Clone$ is complete,
we use \Cref{prop:codensity-limit} to define the "profinite completion of clones".

\begin{definition}
    \AP\label{def:profinite-completion-clones}
    We call ""profinite completion of clones"" the "codensity monad"
    \[
        \begin{tikzcd}[ampersand replacement=\&]
            \FinClone \& \Clone \& \Clone
            \arrow[from=1-1, to=1-2]
            \arrow[""{name=0, anchor=center, inner sep=0}, curve={height=30pt}, from=1-1, to=1-3]
            \arrow[""{name=0p, anchor=center, inner sep=0}, phantom, from=1-1, to=1-3, start anchor=center, end anchor=center, curve={height=30pt}]
            \arrow[""{name=1, anchor=center, inner sep=0}, draw=none, from=1-1, to=1-3]
            \arrow[""{name=1p, anchor=center, inner sep=0}, phantom, from=1-1, to=1-3, start anchor=center, end anchor=center]
            \arrow["{\intro*\ProClone{(-)}}", dashed, from=1-2, to=1-3]
            \arrow[between={0.3}{0.8}, Rightarrow, from=1p, to=0p]
        \end{tikzcd}
    \]
    induced by the inclusion~$\FinClone \to \Clone$.
\end{definition}

\begin{remark}
    We find it instructive to unfold the abstract formulation,
    in terms of limits following \Cref{prop:codensity-limit},
    of the "profinite completion" of clones.
    For any "clone"~$C$ and $n \in \N$,
    an element~$u \in (\ProClone{C})_n$ amounts to a family of functions
    \[
        u_D
        \quad:\quad
        \Clone(C, D)
        \ \longto\
        D_n
        \qquad\text{where $D$ ranges over all "locally finite clones"}
    \]
    which is natural in $D$, i.e. for any $\varphi \in \FinClone(D, D')$,
    the following square commutes
    \[
        \begin{tikzcd}[ampersand replacement=\&]
            {\Clone(C, D)} \&\& {\Clone(C, D')} \\
            {D_n} \&\& {D'_n}
            \arrow["{\varphi \circ (-)}", from=1-1, to=1-3]
            \arrow["{u_D}"', from=1-1, to=2-1]
            \arrow["{u_{D'}}", from=1-3, to=2-3]
            \arrow["{\varphi_n}"', from=2-1, to=2-3]
        \end{tikzcd}
    \]
    In this way,
    we think of $u$ as an oracle with $n$ variables on the clone~$C$ with respect to all "locally finite clones".
\end{remark}

\section{Relating the profinite completions}
\label{sec:profinite-completion-isos}

In this section, we relate the profinite completions of monoids and the ultrafilter monad on sets with the profinite completion of clones.
For this, we develop a general categorical framework which treats both cases at the same time in a generic way.

Our starting point is the observation that the two functors
\[
    \Act
    \quad:\quad
    \Mon
    \ \longto\
    \Clone
    \qquad\text{and}\qquad
    \Cst
    \quad:\quad
    \Set
    \ \longto\
    \Clone
\]
have very simple expressions, written in \Cref{def:act} and \Cref{def:cst},
at odds with the left adjoints often came across in universal algebra typically involving infinite sums and quotients.

\AP
Admitting such a simple description is usually a feature of right adjoints,
but none of these two functors are right adjoints themselves.
Indeed, right adjoints preserve terminal objects,
yet applying the functors $\Act$ and $\Cst$ to the terminal monoid and the terminal set
yields "clones", which we write $\Var$ and $\Err$ respectively,
which are not terminal.
More concretely,
for any natural number $n \in \N$ the sets of operations of these two "clones" are
\begin{align*}
    \intro*\Var_n
    \quad:=\quad
    \{v_1, \dots, v_n\}
    \qquad\text{and}\qquad
    \intro*\Err_n
    \quad:=\quad
    \{\bullet, v_1, \dots, v_n\}
\end{align*}
However, this lack of preservation of the terminal object by the functors~$\Act$ and $\Cst$ is in some sense the only obstruction for them to be right adjoints.
We now recall the notion of "parametric right adjoint",
introduced by Street~\cite{Street2000},
which provides a formal ground for these ideas.

\begin{definition}
    \AP\label{def:pra}
    A functor~$G : \C \to \D$,
    where the category~$\C$ has a terminal object~$1_\C$,
    is called a ""parametric right adjoint""
    when the associated functor
    \[
        \begin{matrix}
            \C & \longto     & \D/G(1_{\C})
            \\
            X  & \longmapsto & G(!_X) : GX \to G(1_{\C})
        \end{matrix}
    \]
    is itself a right adjoint, where $\D/Y$ denotes the slice category.

    In particular, a "parametric right adjoint" preserves connected limits,
    as the projection functor~$\D/Y \to \D$ does.
\end{definition}

In order to show that functors are "parametric right adjoints",
we will leverage the power given by the notion of "theory" of \Cref{sec:functorial-approach},
and in particular \Cref{prop:pull-push-theories} which amounts to the construction of left adjoints.
For this, we now describe a construction on "theories" which give an account of taking slices on categories of "models",
which generalizes the classical equivalence of categories
\[
    \Set^I
    \quad\cong\quad
    \Set/I
\]
for any set $I$ seen as a discrete category, see e.g.~\cite[p.~28]{sheavesingeometryandlogic}.

\begin{definition}
    \AP\label{def:coslice-theory}
    For any "theory"~$T$ and "model"~$X : \T \to \Set$,
    we write $\intro*\coslice{X}{\T}$ for the category of elements of $X$,
    whose objects are pairs
    \[
        (n, x)
        \quad \text{where $n$ is an object of $\T$ and $x \in X(n)$}
    \]
    and whose hom-sets are
    \[
        \coslice{X}{\T}\Big((n, x), (n', x')\Big)
        \quad:=\quad
        \{f \in \T(n, n') \mid X(f)(x) = x'\}
    \]
    There is a canonical functor
    \[
        \Mod(\coslice{X}{\T})
        \ \longto\
        \Mod(\T) / X
    \]
    which sends a "model"~$Y : \coslice{X}{\T} \to \Set$
    on the "model"~$Y' : \T \to \Set$ defined as
    \[
        Y'(n)
        \quad:=\quad
        \sum_{x \in X(n)}
        Y(n, x)
    \]
    together with the natural transformation~$Y' \To X$
    whose components are the projections
    \[
        \sum_{x \in X(n)}
        Y(n, x)
        \ \longto\
        X(n)
        \qquad\text{for every $n \in \T$}
    \]
    This yields an equivalence of categories
    \[
        \Mod(\coslice{X}{\T})
        \quad\cong\quad
        \Mod(\T) / X
    \]
\end{definition}

\begin{remark}
    \label{rmk:coslice-slice-duality}
    For any "model"~$X$ of $\T$ which is a representable functor
    \[
        X
        \ :=\
        \T(n, -)
        \quad:\quad
        \T
        \ \longto\
        \Set
    \]
    the category~$\coslice{X}{\T}$ is the coslice of $\T$ on the object~$n$ of $\T$ representing $X$.
    The fact that the coslice is the "theory" of the slice of "models" is in accordance with the general principle that theories and categories of models are dual one to the other.
\end{remark}

\begin{remark}
    \AP\label{rmk:coslice-pullback}
    Equivalently, the category~$\coslice{X}{\T}$ can be described as the pullback
    \[
        \begin{tikzcd}[ampersand replacement=\&]
            {\coslice{X}{\T}} \& \Setpt \\
            \T \& \Set
            \arrow[from=1-1, to=1-2]
            \arrow[from=1-1, to=2-1]
            \arrow["\lrcorner"{anchor=center, pos=0.125}, draw=none, from=1-1, to=2-2]
            \arrow[from=1-2, to=2-2]
            \arrow["X", from=2-1, to=2-2]
        \end{tikzcd}
        \qquad\text{where $\intro*\Setpt$ is the category of pointed sets.}
    \]
\end{remark}

\AP
We now apply the construction of the "theory"~$\coslice{X}{\T}$ to our approach of clones,
by taking $\T$ to be $\TClone$ and the "model"~$X$ to be either the "clone"~$\Var \cong \Act(1_{\Mon})$ or the "clone"~$\Err \cong \Cst(1_{\Set})$.
This is done in the statements \Cref{prop:act-pra} and \Cref{prop:cst-pra},
which amount to direct computations.

\begin{proposition}
    \label{prop:act-pra}
    Under the equivalence of \Cref{def:coslice-theory}
    \[
        \Clone/\Var
        \quad\cong\quad
        \Mod(\coslice{\Var}{\TClone})
    \]
    the canonical functor
    \[
        \Mon
        \ \longto\
        \Clone/\Var
    \]
    is the precomposition by a product-preserving functor
    \[
        \coslice{\Var}{\TClone}
        \ \longto\
        \TMon
    \]
    Therefore, $\Act : \Mon \to \Clone$ is a "parametric right adjoint".
\end{proposition}

\begin{proposition}
    \label{prop:cst-pra}
    Under the equivalence of \Cref{def:coslice-theory}
    \[
        \Clone/\Err
        \quad\cong\quad
        \Mod(\coslice{\Err}{\TClone})
    \]
    the canonical functor
    \[
        \Set
        \ \longto\
        \Clone/\Err
    \]
    is the precomposition by a product-preserving functor
    \[
        \coslice{\Err}{\TClone}
        \ \longto\
        \TSet
    \]
    Therefore, $\Cst : \Set \to \Clone$ is a "parametric right adjoint".
\end{proposition}

A fundamental fact of "parametric right adjoints" is that they preserve connected limits.
This is a vast class of limits, which includes in particular cofiltered limits
which are of prime importance in the theory of profinite structures,
see e.g.~\cite[Chapter~VI]{johnstone}.
In particular, the index categories of the limits used in \Cref{prop:codensity-limit} to compute right Kan extensions are cofiltered
when the functor
\[
    I
    \quad:\quad
    \C_f
    \ \longto\
    \C
\]
is the inclusion of finite models into $\Mod(\T)$, for any theory~$\T$.
Therefore, in the presence of a "parametric right adjoint"~$G$,
we obtain an isomorphism
\[
    G \circ \cody{I}
    \quad\cong\quad
    \Ran_{I} (G \circ I)
\]
This applies to the encodings of monoids and sets as "clones",
in which case we get a family of isomorphisms of clones
\begin{align*}
    \Act \circ \ProMon{(-)}
    \quad & \cong\quad
    \Ran_I (\Act \circ I)
    \qquad\text{for the inclusion $I : \FinMon \to \Mon$}
    \\
    \Cst \circ \ProSet
    \quad & \cong\quad
    \Ran_I (\Cst \circ I)
    \qquad\text{for the inclusion $I : \FinSet \to \Set$}
\end{align*}

Having established these identifications,
we are now in a position to give an abstract criterion of preservation of right Kan extensions.

\begin{theorem}
    \label{thm:abstrac-criterion}
    Consider the data displayed in the diagram
    \[
        \begin{tikzcd}[ampersand replacement=\&]
            {\C_f} \& \C \\
            {\D_f} \& \D
            \arrow["I", from=1-1, to=1-2]
            \arrow[""{name=0, anchor=center, inner sep=0}, "{L_f}", shift left=2, from=1-1, to=2-1]
            \arrow["L", from=1-2, to=2-2]
            \arrow[""{name=1, anchor=center, inner sep=0}, "{R_f}", shift left=2, from=2-1, to=1-1]
            \arrow["J"', from=2-1, to=2-2]
            \arrow["\dashv"{anchor=center, rotate=-180}, draw=none, from=0, to=1]
        \end{tikzcd}
    \]
    where $\C_f$ and $\D_f$ are essentially small,
    $\C$ and $\D$ are complete,
    $L \circ I = J \circ L_f$,
    and the adjunction~$L_f \dashv R_f$ is coreflective
    i.e. the unit $\Id \To R_f \circ L_f$ is a natural isomorphism.
    From the data of a natural bijection between the hom-sets
    \[
        \D(L, J)
        \quad\cong\quad
        \D(\Id, I \circ R_f)
    \]
    we construct a natural isomorphism
    \[
        \Ran_I (L \circ I)
        \quad\cong\quad
        (\Ran_J J) \circ L
    \]
\end{theorem}

\begin{proof}
    We start by remarking that the data of a natural bijection
    \[
        \D(L, J)
        \quad\cong\quad
        \D(\Id, I \circ R_f)
    \]
    can be rephrased as the data of two natural transformations
    \[
        \widetilde{(-)}
        \quad:\quad
        \D(L, J)
        \ \To\
        \D(\Id, I \circ R_f)
        \qquad\text{and}\qquad
        \epsilon
        \quad:\quad
        L \circ I \circ R_f
        \ \To\
        J
    \]
    From this, we build a natural transformation
    \[
        \Ran_I (L \circ I)
        \quad\To\quad
        (\Ran_J J) \circ L
    \]
    whose components, for any object $X$ of $\C$ are the morphisms
    \[
        \lim_{X \to I X_f} L I X_f
        \ \longto\
        \lim_{L X \to J Y_f} L Y_f
    \]
    defined by the universal property of the limit as coming from the family
    \[
        \begin{tikzcd}[ampersand replacement=\&]
            {\displaystyle\lim_{X \to I X_f} L I X_f} \& {L I R_f Y_f} \& {J Y_f}
            \arrow["{\pi_{\tilde{\varphi}}}", from=1-1, to=1-2]
            \arrow["{\epsilon_{Y_f}}", from=1-2, to=1-3]
        \end{tikzcd}
        \qquad\text{for any $\varphi : L X \to J Y_f$}
    \]
    where $\tilde{\varphi} : X \to I R_f Y_f$ is given by the data of assumed in the hypothesis.

    We now build a natural transformation going in the opposite direction
    \[
        (\Ran_J J) \circ L
        \quad\To\quad
        \Ran_I (L \circ I)
    \]
    whose components, for any object $X$ of $\C$ are the morphisms
    \[
        \lim_{L X \to J Y_f} L Y_f
        \ \longto\
        \lim_{X \to I X_f} L I X_f
    \]
    defined by the universal property of the limit as coming from the family
    \[
        \begin{tikzcd}[ampersand replacement=\&]
            {\displaystyle\lim_{L X \to J Y_f} L Y_f} \&\& {L I X_f}
            \arrow["{\pi_{L \psi}}", from=1-1, to=1-3]
        \end{tikzcd}
        \qquad\text{for any $\psi : X \to I X_f$}
    \]
    given that $L \psi : L X \to L I X_f = J L_f X_f$.
    We then compute that these two natural transformations are inverse one of the other, hence providing the natural isomorphism.
\end{proof}

Given that $\Act$ sends finite monoids to "locally finite clones",
and $\Cst$ sends finite sets to "locally finite clones",
we can apply \Cref{thm:abstrac-criterion} to get the two following theorems.

\begin{theorem}
    \label{thm:act-profinite-iso}
    For any monoid~$M$, we have an isomorphism of clones
    \[
        \ProClone{\Act(M)}
        \quad\cong\quad
        \Act\Big(\ProMon{M}\Big)
    \]
\end{theorem}

\begin{theorem}
    \label{thm:cst-profinite-iso}
    For any set~$X$, we have an isomorphism of clones
    \[
        \ProClone{\Cst(X)}
        \quad\cong\quad
        \Cst\Big(\ProSet X\Big)
    \]
\end{theorem}

\section{Profinite trees and parametricity}
\label{sec:profinite-trees}
In the case of monoids,
a profinite word over the alphabet~$\Sigma$
is an element of the monoid~$\ProMon{\Sigma^*}$,
computed as the profinite completion
of the monoid of finite words~$\Sigma^*$.
In this section,
we extend this idea to the setting of trees through the following definition.

\begin{definition}
    \AP\label{def:profinite-tree}
    For any "ranked alphabet"~$\Sigma$ and $n \in \N$,
    a ""profinite tree"" with $n$ variables is an element of $\ProClone{\Tree(\Sigma)}_n$.
\end{definition}

We wish to give another description of "profinite trees".
For that matter,
we introduce the following notion of "bidefinability",
which amounts to a strong form of parametricity à la Reynolds~\cite{Reynolds83}.

\begin{definition}
    \AP\label{def:bidefinability}
    Let $C$ and $D$ be "clones", $n \in \N$ and $f$ a function
    \[
        f
        \quad:\quad
        \Clone(C, D)
        \ \longto\
        D_n
        \ .
    \]
    For every $t \in C_n$, we say that $f$ is ""defined"" by $t$, or that $t$ \reintro{defines} $f$, if $f$ is the function $u \mapsto u(t)$ which we write $(-) \circ t$ given the bijection $\Clone(\Tree(\{b : n\}), D) \cong D_n$ detailed in \Cref{rmk:endo-automata}.

    For any class of "clones"
    \footnote{The two classes that we will encounter later on are the following:
        \begin{itemize}
            \item the class of "locally finite clones",
            \item the class of "clones" of the form $\Endo(Q)$ where $Q$ is a finite set.
        \end{itemize}
        We emphasize the fact that these classes are essentially small:
        for each class, there is a set containing some of its elements such that any element of the class is isomorphic to one of the set.}
    and any family $u$ of functions
    \[
        u_D
        \quad:\quad
        \Clone(C, D)
        \ \longto\
        D_n
        \qquad\text{for $D$ in that class}
    \]
    we say that $u$ is ""bidefinable"" if for every two "clones" $D$ and $D'$ in the class, there exists a $t \in C_n$ which "defines" both $u_D$ and $u_{D'}$, i.e.
    \begin{align*}
        u_D
        \  & =\
        (-) \circ t
        \quad:\quad
        \Clone(C, D)
        \ \longto\
        D_n
        \\
        u_{D'}
        \  & =\
        (-) \circ t
        \quad:\quad
        \Clone(C, D')
        \ \longto\
        D'_n
    \end{align*}
\end{definition}

\begin{proposition}
    \ZAP\label{prop:bidefinable-implies-natural}
    Let $C$ be a "clone", $n \in \N$ and a fixed class of "clones". If $u$ is a family of functions
    \[
        u_D
        \quad:\quad
        \Clone(C, D)
        \ \longto\
        D_n
        \qquad\text{for all $D$ in that class}
    \]
    which is "bidefinable", then it is natural.
\end{proposition}

\begin{proof}
    Let $D, D'$ be two "clones" of the fixed class and $\varphi : D \to D'$ be a "clone morphism". By "bidefinability" of $u$, there exists $t \in C_n$ which "defines" $u_D$ and $u_{D'}$. Then, the commutativity of the square
    \[
        \begin{tikzcd}[ampersand replacement=\&]
            {\Clone(C, D)} \&\& {\Clone(C, D')} \\
            {D_n} \&\& {D'_{n}}
            \arrow["{\varphi \circ (-)}", from=1-1, to=1-3]
            \arrow["{(-) \circ t}"', from=1-1, to=2-1]
            \arrow["{(-) \circ t}", from=1-3, to=2-3]
            \arrow["{\varphi \circ (-)}"', from=2-1, to=2-3]
        \end{tikzcd}
    \]
    comes from associativity of composition of functions, so $u$ is natural.
\end{proof}

The converse of \Cref{prop:bidefinable-implies-natural} does not hold in general. However, it does in the case of "profinite trees", i.e. elements of the profinite completion of a free "clone".

\begin{lemma}
    \ZAP\label{lem:local-definability}
    Let $C$ be a "clone" and $n \in \N$.

    For any "locally finite clone" $D$, morphism $p \in \Clone(C, D)$
    and element $u \in \Pro{C}_n$,
    there exists a $t \in C_n$ such that
    \[
        u_D(p)
        \ =\
        p \circ t
        \ .
    \]
\end{lemma}

\begin{proof}
    Let $C$, $n$, $u$, $D$ and $p$ as in the statement. We note $\Image(p)$ for the image of $p$, i.e. the defined as
    \[
        \Image(p)_m
        \quad:=\quad
        \{p_m(t) : t \in C_n\}
        \qquad\text{for all $m \in \N$}
    \]
    which permits factorizing $p : C \to D$ as
    \[
        \begin{tikzcd}[ampersand replacement=\&]
            C \& {\Image(p)} \& D
            \arrow["{p_{\Image}}", from=1-1, to=1-2]
            \arrow["\iota", from=1-2, to=1-3]
        \end{tikzcd}
    \]
    By naturality of $u$, we thus get that
    \[
        u_D(p)
        \ =\
        \iota \circ (u_{\Image(p)}(p_{\Image}))
    \]
    and hence there exists $t \in C_n$ such that $u_D(p) = p_n(t)$.
\end{proof}

\begin{proposition}
    \AP\label{prop:free-clones-bidefinable}
    Let $\Sigma$ be any "ranked alphabet".
    "Profinite trees" on $\Sigma$ are "bidefinable".
\end{proposition}

\begin{proof}
    Let us write $C$ for the "clone"~$\Tree(\Sigma)$.
    Let $u$ be a "profinite tree" on $\Sigma$ with $n$ variables, i.e. a natural family of functions
    \[
        u_D
        \quad:\quad
        \Clone(\Tree(\Sigma), D)
        \ \longto\
        \Clone(\Tree(\{b : n\}), D)
    \]
    where $D$ ranges over all "locally finite clones". We pick two "locally finite clones" $D$ and $D'$. Using the fact that $\FinClone$ is a sub-cartesian category of $\Clone$, we consider the locally finite clone $D^\times$ defined as
    \[
        D^\times
        \quad:=\quad
        D^{\Clone(C, D)} \times D'^{\Clone(C, D')}
    \]
    which hence induces a bijection between $\Clone(C, D^\times)$ and
    \[
        {\Clone(C, D)}^{\Clone(C, D)} \times {\Clone(C, D')}^{\Clone(C, D')}
    \]
    and we note $\pgen \in \Clone(C, D^\times)$ for the element corresponding to the pair $(\Id, \Id)$ through this bijection.
    By \Cref{lem:local-definability} applied to $\pgen$, we get the existence of $t \in C_n$ such that $u_{D^\times}(\pgen) = \pgen \circ t$. For any $p \in \Clone(C, D)$ and $p' \in \Clone(C, D')$, the naturality of $u$ gives that
    \begin{align*}
        u_{D}(p)
        \  & =\
        u_{D}(\pi_p \circ \pi_1 \circ \pgen)
        \\&=\
        (\pi_p \circ \pi_1)_n(u_{D^\times}(\pgen))
        \\&=\
        (\pi_p \circ \pi_1)_n(\pgen \circ t)
        \\&=\
        (\pi_p \circ \pi_1)_n(\pgen) \circ t
        \\&= \
        p \circ t
    \end{align*}
    which proves that $t$ "defines" $u_D$. In the same way, $t$ "defines" $u_{D'}$, which proves that $u$ is "bidefinable".
\end{proof}

The proof of \Cref{prop:free-clones-bidefinable} can be used to show that
"profinite trees" are actually multidefinable,
for this meaning the same as "bidefinability" from \Cref{def:bidefinability} but taking a finite number of components of the family.
For that, it suffices in the proof to consider the clone
\[
    D^\times
    \quad:=\quad
    \prod_{1 \le i \le n}
    \Clone(C, D_i)^{\Clone(C, D_i)}
\]
which is \kl[locally finite clone]{locally finite} as a finite product of "locally finite clones".

\begin{remark}
    \label{rmk:orthorel}
    We consider the binary relation $\orthorel$ on "clones" such that
    \[
        C \orthorel D
        \qquad\text{if and only if}\qquad
        \text{the set $\Clone(C, D)$ is finite}
    \]
    Modulo set-theoretic matters,
    this relation~$\orthorel$ induces two assignments $\orthoright{(-)}$ and $\ortholeft{(-)}$ on classes of "clones", defined as
    \begin{align*}
        \orthoright{\mathcal{F}}
        \quad & :=\quad
        \{\text{"clone"~$D$} \mid \text{for all } C \in \mathcal{F}, \text{ we have } C \orthorel D\}
        \\
        \ortholeft{\mathcal{G}}
        \quad & :=\quad
        \{\text{"clone"~$C$} \mid \text{for all } D \in \mathcal{G}, \text{ we have } C \orthorel D\}
    \end{align*}
    which assemble into a Galois connection on classes of "clones".
    The class~$\orthoright{\mathcal{F}}$ is always stable by subobjects and finite limits,
    and the class $\ortholeft{\mathcal{G}}$ is always stable by quotients and finite colimits.
    Given that free "clones" of the form $\Tree(\Sigma)$ correspond to representable "models"~$\TClone(\Sigma, -)$,
    by the Yoneda lemma we have that
    \[
        \orthoright{\{\text{free "clones"}\}}
        \quad=\quad
        \{\text{"locally finite clones"}\}
    \]
    The proof of \Cref{prop:free-clones-bidefinable} applies verbatim
    to the case where the the free "clone"~$\Tree(\Sigma)$ is replaced
    by any "clone"~$C$ in the class
    \[
        \ortholeft{\{\text{"locally finite clones"}\}}
    \]
    Many of them are not built out of free ones by finite colimits and quotients.
    For example, the monoid $\Q$ of rational numbers equipped with their addition
    induces the "clone"~$\Act(\Q)$ which belongs to that class,
    given that for any "locally finite clone"~$D$,
    we have
    \[
        \Clone(\Act(\Q), D)
        \quad\cong\quad
        \Mon(\Q, D_1)
        \quad=\quad
        \{x \mapsto v_1\}
    \]
\end{remark}

\section{Profinite trees and the profinite \pdflambda-calculus}
\label{sec:profinite-lambda-terms}

\AP
In this section,
we recall the notion of "profinite $\lambda$-term" introduced in~\cite{entics:12280}
and we establish that the "profinite trees" on a given "ranked alphabet"~$\Sigma$,
introduced in \Cref{sec:profinite-trees},
are exactly the "profinite $\lambda$-terms" of type~$\Church{\Sigma}$, i.e. of the type
\[
    \intro*\Church{\Sigma}
    \quad:=\quad
    (\tyo^{n_1} \To \tyo) \To \dots \To (\tyo^{n_l} \To \tyo) \To \tyo
\]
for any "ranked alphabet"~$\Sigma := \{a_1 : n_1, \dots, a_l : n_l\}$.
This is moreover a homeomorphism of Stone spaces which preserves the clone structure.

We start by recalling the notion of "profinite $\lambda$-term" of type~$\Church{\Sigma}$, taking advantage of their reformulation in terms of parametricity developed in~\cite[Chapter~8]{moreau:tel-05428993}.
For this, we use the notion of "bidefinability" introduced in \Cref{def:bidefinability}.

\begin{definition}
    \AP\label{def:profinite-lambda-term}
    For any "ranked alphabet"~$\Sigma$ and natural number~$n \in \N$,
    a ""profinite $\lambda$-term""~$\theta$ of type $\tyo^n \To \Church{\Sigma}$ is a "bidefinable" family of set-theoretic functions
    \[
        \theta_Q
        \quad:\quad
        \Clone(\Tree(\Sigma), \Endo(Q))
        \ \longto\
        \Endo(Q)_n
    \]
    where $Q$ ranges over all finite sets.
    We write $\intro*\ProTm(\tyo^n \To \Church{\Sigma})$ for the set of "profinite $\lambda$-terms" of that type.
\end{definition}

\begin{remark}
    \label{rmk:church-tyo}
    For any "ranked alphabet"~$\Sigma$, the type
    \[
        \tyo^n \To \Church{\Sigma}
        \qquad\text{is itself of the form}\qquad
        \Church{\Gamma}
    \]
    where $\Gamma$ is the extended "ranked alphabet" defined as $\Sigma \sqcup \{v_1 : 0, \dots, v_n : 0\}$.
    The notion of "profinite $\lambda$-term" given in \Cref{def:profinite-lambda-term},
    parametrized by a "ranked alphabet",
    yields the same objects with $\Sigma$ and $n$, and with $\Gamma$ and $0$.
\end{remark}

We have proven in \Cref{prop:free-clones-bidefinable} that "profinite trees"
are the "bidefinable families" over the class of all "locally finite clones".
"Profinite $\lambda$-terms" as in \Cref{def:profinite-lambda-term}
are "bidefinable families" over the class of all "clones"~$\Endo(Q)$ where $Q$ ranges over finite sets.
As a consequence,
it is straightforward to associate
to any "profinite tree"~$u$
a "profinite $\lambda$-term"~$\restrict(u)$.

\begin{definition}
    \ZAP\label{def:trees-to-terms}
    Let $\Sigma$ be a "ranked alphabet" and $n \in \N$.
    For any "profinite tree"~$u$ over $\Sigma$ with $n$ variables,
    we define a "profinite $\lambda$-term"
    \[
        \theta
        \quad:=\quad
        \restrict(u)
    \]
    of type $\tyo^n \To \Church{\Sigma}$, whose components are the functions
    \[
        \theta_Q
        \quad:=\quad
        u_{\Endo(Q)}
        \qquad\text{where $Q$ ranges over all finite sets.}
    \]
    This defines a family of set-theoretic functions
    \[
        \intro*\restrict_n
        \quad:\quad
        \Pro{\Tree(\Sigma)}_n
        \ \longto\
        \ProTm(\tyo^n \To \Church{\Sigma})
        \qquad\text{for every $n \in \N$}
    \]
    which defines a "clone morphism" from "profinite trees" to "profinite $\lambda$-terms".
\end{definition}

We now go the other way around
i.e. from "profinite $\lambda$-terms" to "profinite trees".
to achieve this,
we consider a Cayley morphism for "clones".
Through the correspondence between clones and finitary monads,
this morphism to the fact that
each set $C_n$ of a clone $C$ carries a structure of algebra,
as it indeed is the free algebra on $n$ generators.

\begin{definition}
    \ZAP\label{def:cayley-clone}
    Let $D$ be a "clone" and $n \in \N$.
    We write $\cay^n$ for the clone morphism
    \[
        \intro*\cay^n
        \quad:\quad
        D
        \ \longto\
        \Endo(D_n)
    \]
    whose components are the functions
    \[
        \cay^n_m
        \quad:\quad
        \begin{matrix}
            D_m
             & \longto     &
            \Endo(D_n)_m
            \\
            x
             & \longmapsto &
            \begin{pmatrix}
                (D_n)^m
                 & \to     &
                D_n
                \\
                (y_1, \dots, y_m)
                 & \mapsto &
                x[y_1, \dots y_m]
            \end{pmatrix}
        \end{matrix}
    \]
    obtained by currying the "substitution"~$(-_1)[-_2]$ of the "clone"~$D$.
\end{definition}

To obtain a "profinite $\lambda$-term" from a "profinite tree",
we have seen in \Cref{def:trees-to-terms} that we only need to restrict a "profinite tree"
from all "locally finite clones" to those of the form $\Endo(Q)$ for a finite set~$Q$.
To go the other way around, we define the following function,
dual to the morphism~$\cay^n$.

\begin{definition}
    \AP\label{def:appvar}
    For any "clone"~$D$ and natural number~$n \in \N$,
    we consider the following set-theoretic function defined by application of "variables":
    \[
        \intro*\appvar
        \quad:\quad
        \begin{matrix}
            \Endo(D_n)_n
             & \longto     &
            D_n
            \\
            f
             & \longmapsto &
            f(v_1, \dots, v_n)
        \end{matrix}
    \]
    Using the right unitality axiom described in \Cref{def:clone},
    the function $\appvar$ appears as a retraction of the component function
    \[
        \cay^n_n
        \quad:\quad
        \begin{matrix}
            D_n
             & \longto     &
            \Endo(D_n)_n
            \\
            x
             & \longmapsto &
            \Big((y_1, \dots, y_m) \mapsto x[y_1, \dots y_m]\Big)
        \end{matrix}
    \]
\end{definition}

\begin{proposition}
    \ZAP\label{prop:restrict-injective}
    For any "ranked alphabet"~$\Sigma$ and natural number~$n \in \N$,
    the function
    \[
        \restrict_n
        \quad:\quad
        \ProClone{\Tree(\Sigma)}_n
        \ \longto\ \ProTm(\tyo^n \To \Church{\Sigma})
    \]
    is injective.
\end{proposition}

\begin{proof}
    Let $u$ and a "profinite tree".
    By naturality, for any "locally finite clone"~$D$ the following diagram commutes:
    \[
        \begin{tikzcd}[ampersand replacement=\&]
            {\Clone(\Tree(\Sigma), D)} \&\& {\Clone(\Tree(\Sigma), \Endo(D_n))} \\
            {D_n} \&\& {\Endo(D_n)_n}
            \arrow["{\cay^n \circ (-)}", from=1-1, to=1-3]
            \arrow["{u_D}"', from=1-1, to=2-1]
            \arrow["{u_{\Endo(D_n)}}", from=1-3, to=2-3]
            \arrow["{\cay^n_n}"', from=2-1, to=2-3]
        \end{tikzcd}
    \]
    Composing by the retraction~$\appvar$ of $\cay^n_n$,
    we obtain the following commutative diagram
    \[
        \begin{tikzcd}[ampersand replacement=\&]
            {\Clone(\Tree(\Sigma), D)} \&\& {\Clone(\Tree(\Sigma), \Endo(D_n))} \\
            {D_n} \&\& {\Endo(D_n)_n}
            \arrow["{\cay^n \circ (-)}", from=1-1, to=1-3]
            \arrow["{u_D}"', from=1-1, to=2-1]
            \arrow["{u_{\Endo(D_n)}}", from=1-3, to=2-3]
            \arrow["\appvar", from=2-3, to=2-1]
        \end{tikzcd}
    \]
    which shows that $u$ is determined by the subfamily of its components~$u_{\Endo(Q)}$ for every finite set~$Q$.
    As a consequence, the function $\restrict_n$ is injective.
\end{proof}

\begin{proposition}
    \ZAP\label{prop:restrict-surjective}
    For any "ranked alphabet"~$\Sigma$, natural number~$n \in \N$
    and "profinite $\lambda$-term"~$\theta$ of type $\tyo^n \To \Church{\Sigma}$,
    the family of functions $u_D$, for a "locally finite clone"~$D$, defined as the composition
    \[
        \begin{tikzcd}[ampersand replacement=\&]
            {\Clone(\Tree(\Sigma), D)} \&\& {\Clone(\Tree(\Sigma), \Endo(D_n))} \\
            {D_n} \&\& {\Endo(D_n)_n}
            \arrow["{\cay^n \circ (-)}", from=1-1, to=1-3]
            \arrow["{u_D}"', dashed, from=1-1, to=2-1]
            \arrow["{\theta_{D_n}}", from=1-3, to=2-3]
            \arrow["\appvar", from=2-3, to=2-1]
        \end{tikzcd}
    \]
    assembles into a "profinite tree"~$u$ over $\Sigma$ with $n$ variables.
    In other words, the function
    \[
        \restrict_n
        \quad:\quad
        \ProClone{\Tree(\Sigma)}_n
        \ \longto\ \ProTm(\tyo^n \To \Church{\Sigma})
    \]
    is surjective.
\end{proposition}

\begin{proof}
    Let $\theta$ be a "profinite $\lambda$-term" of type~$\tyo^n \To \Church{\Sigma}$.
    We need to show that the families of $u_D$,
    where $D$ ranges over all "locally finite clones", is "bidefinable".
    We remark that for any "locally finite clone"~$D$,
    if $\theta_{D_n}$ is "defined" by $t$, then so is $u_D$.
    Therefore, as for any "locally finite clones"~$D$ and~$D'$
    there exists $t$ that "defines" both $\theta_{D_n}$ and $\theta_{D'_n}$,
    then the same $t$ "defines" both $u_D$ and $u_{D'}$.
    This shows that the family~$u$ of all the $u_D$ is "bidefinable".
\end{proof}

As a "profinite completion", the "clone"~$\ProClone{\Tree(\Sigma)}$ naturally carries topologies of Stone spaces on each of the set $\ProClone{\Tree(\Sigma)}_n$ for $n \in \N$.
Moreover, the sets $\ProTm(\tyo^n \To \Church{\Sigma})$ also carry a topology making them Stone spaces,
see~\cite[\S~3]{entics:12280}.
Taking the two topologies into account, we then obtain the following theorem:

\begin{theorem}
    \ZAP\label{thm:isomorphism-theorem}
    For any "ranked alphabet"~$\Sigma$, the "clone morphism"
    \[
        \restrict
        \quad:\quad
        \ProClone{\Tree(\Sigma)}
        \ \longto\
        \ProTm\Big(\tyo^{(-)} \To \Church{\Sigma}\Big)
    \]
    is an isomorphism of "clones" "enriched" in Stone spaces.
\end{theorem}

\begin{proof}
    Taken together, \Cref{prop:restrict-injective} and \Cref{prop:restrict-surjective} imply that all the components of the clone morphism~$\restrict$ are bijections.
    As they also preserve "variables" and "substitutions",
    we obtain that the whole morphism is an isomorphism of "clones".
    As Stone spaces are compact Hausdorff, to show that each component is a homeomorphism, it suffices to show that it is continuous.

    The topology on $\ProClone{\Tree(\Sigma)}_n$ is the topology induced by the inclusion in the product
    \[
        \ProClone{\Tree(\Sigma)}_n
        \quad\subseteq\quad
        \prod_{p : \Tree(\Sigma) \to D} D_n
    \]
    More concretely, it has a sub-base whose open sets are of the following form, for $D$ is a "locally finite clone", $p \in \Clone(\Tree(\Sigma), D)$ and $S$ a subset of $D_n$,
    \[
        U_{p, S}
        \quad:=\quad
        \left\{
        u \in \ProClone{\Tree(\Sigma)}_n
        \mid
        u_C(p) \in S
        \right\}
    \]
    The topology on $\ProTm(\tyo^n \To \Church{\Sigma})$ is given by the inclusion in the product space
    \[
        \ProTm(\tyo^n \To \Church{\Sigma})
        \quad\subseteq\quad
        \prod_{Q} \Big(\Endo(Q)_n\Big)^{\Clone(\Tree(\Sigma), \Endo(Q))}
    \]
    More concretely, it has a sub-base whose open sets are of the form, for $Q$ a finite set and $T$ a subset of $\left(\Endo(Q)_n\right)^{\Clone(\Tree(\Sigma), \Endo(Q))}$,
    \[
        V_{Q, T}
        \quad:=\quad
        \left\{
        \theta \in \ProTm(\tyo^n \To \Church{\Sigma})
        \mid
        \theta_Q \in T
        \right\}
        \ .
    \]
    We now verify that, for $u \in \ProClone{\Tree(\Sigma)}_n$, for $Q$ any finite set and for $T \subseteq \left(\Endo(Q)_n\right)^{\Clone(\Tree(\Sigma), \Endo(Q))}$,
    we have that
    \begin{align*}
        \restrict_n(u) \in V_{Q, T}
        \quad\iff & \quad
        u_{\Endo(Q)} \in V_{Q, T}
        \\\iff&\quad
        \exists f \in T, u_{\Endo(Q)} = f
        \\\iff&\quad
        \exists f \in T, \forall p \in \Clone(\Tree(\Sigma), \Endo(Q)), u_{\Endo(Q)}(p) = f(p)
        \\\iff&\quad
        \exists f \in T, \forall p \in \Clone(\Tree(\Sigma), \Endo(Q)), u \in U_{p, \{f(p)\}}
        \\\iff&\quad
        u \in \bigcup_{f \in T} \bigcap_{p \in \Clone(\Tree(\Sigma), \Endo(Q))} U_{p, \{f(p)\}}
        \ .
    \end{align*}
    Therefore, $\restrict_n^{-1}(V_{Q, T})$ is a finite union of finite intersections of $U_{p, S}$, so it is an open set of $\ProClone{\Tree(\Sigma)}_n$. This proves that
    \[
        \restrict_n
        \quad:\quad
        \ProClone{\Tree(\Sigma)}_n
        \ \longto\
        \ProTm(\tyo^n \To \Church{\Sigma})
    \]
    is continuous, so it is a homeomorphism.
    Therefore, $\restrict$ is an isomorphism of clones enriched in Stone spaces.
\end{proof}

\section{Conclusion}

In this paper, we consider "clones" as an appropriate notion of algebra to treat finite ranked trees with sharing and their run inside of deterministic, bottom-up automata.
We define the "profinite completion" of "clones",
and we develop an abstract criterion based on parametric right adjoints to show that it generalizes the profinite completion of monoids and the ultrafilter monad.
When applied to free clones, this profinite completion yields a notion of profinite tree,
which we show to coincide with a fragment of the recently introduced profinite $\lambda$-calculus.

As future work, we want to develop the theory of pseudovarieties of locally finite clones and establish an analog Reiterman's theorem~\cite{Reiterman1982}
using our newly introduced notion of profinite tree.
We are also interested in other applications of the techniques based on parametric right adjoints,
in the algebraic setting of the "theories" considered here,
to codensity monads.

\bibliography{biblio}

\appendix

\end{document}